\documentclass[conference]{IEEEtran}

\usepackage[utf8]{inputenc}

\usepackage{microtype} 
\usepackage{ragged2e} 

\usepackage{flushend}

\usepackage{cite}
\usepackage[pdftex]{graphicx}

\usepackage{amssymb,amsmath,amsthm} 
\usepackage{array}
\usepackage{fixltx2e}
\usepackage{url}

\newtheorem{lemma}{Lemma}
\newcommand{\E}{\mathbb{E}}
\renewcommand{\P}{\mathbb{P}}
\newcommand{\R}{\mathbb{R}}
\newcommand{\K}{\mathbf{K}}
\renewcommand{\d}{\,\mathrm{d}}

\makeatletter
\def\@IEEEproof[#1]{\par\noindent{\itshape #1: }}
\makeatother

\begin{document}
\IEEEoverridecommandlockouts

\title{\hbox{\Huge \!\!\!\!Optimal Transportation to the Entropy-Power Inequality}} 
\author{\IEEEauthorblockN{Olivier Rioul}
\IEEEauthorblockA{LTCI, T\'el\'ecom ParisTech, Universit\'e Paris-Saclay, 75013 Paris, France\\
Email: olivier.rioul@telecom-paristech.fr}
}

\maketitle

\begin{abstract}
We present a simple proof of the entropy-power inequality using an optimal transportation argument which takes the form of a simple change of variables. The same argument yields a reverse inequality involving a conditional differential entropy which has its own interest. It can also be generalized in various ways. The equality case is easily captured by this method and the proof is formally identical in one and several dimensions. 
\end{abstract}

\section{Introduction}

The entropy-power inequality gives a lower bound on the differential entropy of a sum of independent random vectors in terms of their individual differential entropies, and is perhaps the most fascinating inequality that was stated by Shannon in his 1948 seminal paper~\cite{shannon48}. 
To simplify the presentation we assume, without loss of generality, that all considered random vectors have \emph{zero mean} and we first restrict ourselves to real-valued random \emph{variables} in one dimension.

Letting $P(X)=\E\{X^2\}$ be the (average) power of a random variable $X$, Shannon defined the \emph{entropy-power} $N(X)$ as the \emph{power} of a Gaussian random variable $X^*$ having the same \emph{entropy} as $X$. He argued~\cite[\S~21]{shannon48} that for continuous random variables it is more convenient to work with the entropy-power $N(X)$ than with the differential entropy $h(X)$.

By the well-known formula $h(X^*)=\frac{1}{2}\log \bigl(2\pi e P(X^*)\bigr)$ of the entropy of the Gaussian $X^*$, the closed-form expression of $N(X)=P(X^*)$ when $h(X^*)=h(X)$ is
\begin{equation}\label{EP}
N(X)= \frac{e^{2 h(X)} }{2\pi e}
\end{equation}
which is essentially $e$ to the \emph{power} twice the \emph{entropy} of~$X$, also the ``entropy power'' of $X$ in this sense. Since the Gaussian maximizes entropy for a given power: $h(X) \leq \frac{1}{2}\log \bigl(2\pi e P(X)\bigr)$, the entropy-power does not exceed the actual power: $N(X)\leq P(X)$ with equality if and only if $X$ is Gaussian.

A basic property of the entropy-power is the \emph{scaling} property. The power of a scaled random variable is given by $P(aX)=a^2 P(X)$, and the same property holds for the entropy-power:
\begin{equation}\label{EPscaling}
N(aX)=a^2N(X) 
\end{equation}
thanks to the  well-known scaling property of the entropy:
\begin{equation}\label{scaling}
h(aX)=h(X)+\log a \qquad (a>0).
\end{equation}

For any two \emph{independent} continuous random variables $X$ and $Y$, the power of the sum equals the sum of the individual powers: $P(X+Y)=P(X)+P(Y)$ and clearly the same relation holds for the entropy-power in the case of Gaussian variables. For non-Gaussian variables, however, the entropy-power of the sum exceeds the sum of the individual entropy-powers:
\begin{equation}\label{EPI}
N(X+Y) \geq N(X) + N(Y) 
\end{equation}
where equality holds only if $X$ and $Y$ are Gaussian. This is the celebrated entropy-power inequality (EPI) as stated by Shannon. It is remarkable that Shannon had the intuition of this inequality since it turns out to be quite difficult to prove. The first rigorous proof is due to Stam~\cite{stam59} more than ten years after Shannon's paper and is quite involved.

Thirty years after Shannon's paper, Lieb~\cite{lieb78} gave a very different proof of an equivalent entropy-power inequality that is more convenient to prove. By the scaling property~\eqref{EPscaling}, one has $N(\sqrt{\lambda} X)=\lambda N(X)$ for any $0<\lambda<1$, and the EPI~\eqref{EPI} is clearly equivalent to
\begin{equation}\label{Nconcave}
N(\sqrt{\lambda} X + \sqrt{1-\lambda} Y)  \geq \lambda N(X) + (1-\lambda) N(Y).
\end{equation}
Taking the logarithm on both sides it follows from the concavity of the logarithm that
\begin{equation}\label{hconcave}
h(\sqrt{\lambda} X + \sqrt{1-\lambda} Y)  \geq \lambda h(X) + (1-\lambda) h(Y).
\end{equation}
Conversely, to prove~\eqref{Nconcave} it is sufficient, by appropriately scaling the variables, to assume that $X$ and $Y$ have the same entropy power $N(X)=N(Y)$, hence the same entropy $h(X)=h(Y)$. In this case, taking the exponential on both sides of~\eqref{hconcave}, the r.h.s. becomes $(e^{2h(X)})^\lambda(e^{2h(Y)})^{1-\lambda}=\lambda e^{2h(X)} + (1-\lambda) e^{2h(Y)}$ which gives~\eqref{Nconcave}. Thus Lieb's restatement~\eqref{hconcave} is equivalent to the EPI. Equality holds in~\eqref{hconcave} if and only if $X$ and $Y$ are Gaussian with the same power.

Both~\eqref{Nconcave} and~\eqref{hconcave} have a nice interpretation~\cite{DemboCoverThomas91}: both the entropy-power $N$ and the entropy $h$ are \emph{concave} under the \emph{power-preserving} combination $\sqrt{\lambda} X + \sqrt{1-\lambda} Y$. That linear combination is power-preserving because if $X$ and $Y$ have the same power $P$, then $\sqrt{\lambda} X + \sqrt{1-\lambda} Y$ also has the same power $P$.

All available proofs of the EPI~\eqref{hconcave} can be seen as either variants of Stam's proof using a Gaussian perturbation argument (where the entropies are differentiated with respect to the power of an additive Gaussian noise), or variants of Lieb's proof using sharp inequalities from functional analysis such as Young's convolutional inequality (where the EPI is obtained as a limit case). In this paper, we present a new proof from~\cite{Rioul17} using a \emph{transportation} argument in which the Gaussian distribution is ``transported'' to another probability distribution by a simple change of variable. The idea is to relate~\eqref{hconcave} to the case of equality: let $X^*$, $Y^*$ be independent Gaussian 
with the same power, so that
\begin{equation}\label{hconcave*}
h(\sqrt{\lambda} X^* + \sqrt{1-\lambda} Y^*)  = \lambda h(X^*) + (1-\lambda) h(Y^*)
\end{equation}
A \emph{transportation} from $X^*$ to $X$, and similarly from $Y^*$ to $Y$, can be made to compare $h(X)$ to $h(X^*)$, $h(Y)$ to $h(Y^*)$, and also $h(\sqrt{\lambda} X + \sqrt{1-\lambda} Y)$ to $h(\sqrt{\lambda} X^* + \sqrt{1-\lambda} Y^*)$. This is described in the following section.

\section{Ingredients}

Hereafter we assume that the considered random variables have continuous and positive densities. This assumption can be made without loss of generality (see~\cite{Rioul17} for details). It follows that all considered cumulative distribution functions are continuously differentiable and (strictly) increasing.

The following lemma is the ``not Gaussian to Gaussian'' lemma~1 used in~\cite{RioulCosta16}:
\begin{lemma}[Transportation]\label{transport}
There exists an increasing function~$T$ such that $T(X^*)$ has the same distribution as~$X$. 
\end{lemma}
\begin{IEEEproof}
Let $F_X$ denote the cumulative distribution function of $X$. Then  
$\P\{X\leq x\} = F_X(x)=F_{X^*}\bigl( F^{-1}_{X^*} (F_X(x))\bigr) =\P\{X^*\leq F^{-1}_{X^*} (F_X(x)\}= \P\{F_{X}^{-1}\bigl(F_{X^*}(X^*)\bigr)\leq x\}$ which proves the lemma with $T=F_{X}^{-1}\circ F_{X^*}$.
\end{IEEEproof}
Notice that the lemma is well-known when $X^*$ is uniformly distributed, to justify the inverse transform sampling method.

This function $T$ is sometimes referred to an ``optimal transport''~\cite{Villani08} because it solves a Monge-Kantorovitch transportation problem of the type:
$$
\min_{\substack{(X,X^*)\\X\sim p_X, X^*\sim p_{X^*}}} \sqrt{\mathbb{E}\{(X-X^*)^2\}} 
$$
where the marginal densities are fixed and the minimisation of the transportation cost is done on the joint distribution.
The resulting minimum is known as the Wasserstein distance $W_2(X,X^*)$. Thus $X=T(X^*)$ is the random variable which is maximally correlated to $X^*$ for fixed marginals; this is a restatement of a well-known Hardy-Littlewood rearrangement inequality and can be generalized to other convex cost functions. This type of optimality was used in~\cite{PolyanskiyWu16} to prove Costa's corner point conjecture for the Gaussian interference channel (see also~\cite{Rioul17b}). However, we shall \emph{not} need such an optimality property here.

By Lemma~\ref{transport}, to prove the EPI we can always assume that $X=T(X^*)$ using transport $T$, and similarly $Y=U(Y^*)$ using another transport $U$. Thus the EPI can be restated in terms of the Gaussian variables $X^*,Y^*$ as
\begin{multline}\label{hconcavegauss}
h(\sqrt{\lambda} T(X^*) + \sqrt{1-\lambda} U(Y^*))  \\\geq \lambda h(T(X^*)) + (1-\lambda) h(U(Y^*)).
\end{multline}

We have the following well-known lemma (also used in~\cite{RioulCosta16}).
\begin{lemma}[Change of Variable in the Entropy]\label{changevariable}
\begin{equation}
h(T(X))= h(X) + \E \{ \log T'(X) \} 
\end{equation}
where $T'>0$ denotes the derivative of $T$.
\end{lemma}
\noindent For  linear $T(x)=ax$ we recover the scaling property~\eqref{EPscaling}. The general proof is similar:
\vspace*{-\medskipamount}
\begin{proof}
Make the change of variable $p_{T(X)}(T(x))\d T(x) = p_{X}(x) \d x$ in the expression of the entropy
$h(T(X))=-\E \{\log p_{T(X)}(T(X))\}=-\E\{\log (p_{X}(X)/T'(X) ) \}=h(X) \linebreak + \E \{ \log T'(X) \}$.
\end{proof}
Lemma~\ref{changevariable} allows one to evaluate the differences $h(T(X^*))-h(X^*)$ and $h(U(Y^*))-h(Y^*)$.
However, the remaining terms $h(\sqrt{\lambda} T(X^*) + \sqrt{1-\lambda} U(Y^*))$ and $h(\sqrt{\lambda} X^* + \sqrt{1-\lambda} Y^*)$ cannot be compared directly because two variables are involved instead of one.
However one variable can be fixed by conditioning and an extended version of Lemma~\ref{changevariable} can be used:
\begin{lemma}[Change of Variable in the Conditional Entropy]\label{changevariablecond}
\begin{equation}
h(T_Y(X)|Y)= h(X|Y) + \E \{ \log T'_Y(X) \}.
\end{equation}
\end{lemma}
\begin{proof}
By Lemma~\ref{changevariable}, we have $h(T_Y(X)|Y=y)= h(X|Y=y) + \E \{ \log T'_Y(X)|Y=y \}$ for a fixed value $Y=y$. The result follows by taking the expectation over $Y$.
\end{proof}
Using these ingredients, a simple proof of the EPI is obtained as shown in the next section.

\section{A Simple Proof of the EPI}

From Lemma~\ref{transport} we can assume that $X=T(X^*)$ using transport~$T$ and $Y=U(Y^*)$ using transport~$U$. By Lemma~\ref{changevariable},
\begin{equation}
\begin{aligned}
h(X)&=h(X^*) + \E \{ \log T'(X^*) \}\\
h(Y)&=h(Y^*) +  \E \{ \log U'(Y^*) \}.
\end{aligned}\label{transportxy} 
\end{equation}
It remains to compare $h(\sqrt{\lambda} X + \sqrt{1-\lambda} Y) =h(\sqrt{\lambda} T(X^*) + \sqrt{1-\lambda} U(Y^*))$ to $h(\sqrt{\lambda} X^* + \sqrt{1-\lambda} Y^*)$, which is the entropy if the Gaussian variable $\widetilde{X}=  \sqrt{\lambda} \,X^* + \sqrt{1-\lambda}\, Y^*$. Two independent variables are involved in the expression $\sqrt{\lambda} T(X^*) + \sqrt{1-\lambda} U(Y^*)$ which does not depend on $\widetilde{X}$ alone, but rather on the two variables $(\widetilde{X}, \widetilde{Y})$ obtained by rotation from $(X^*,Y^*)$:
\begin{equation}
\begin{pmatrix}
\widetilde{X} \\ \widetilde{Y} 
\end{pmatrix} = 
\begin{pmatrix}
\sqrt{\lambda} & \sqrt{1-\lambda} \\
-\sqrt{1-\lambda} & \sqrt{\lambda}
\end{pmatrix}
\begin{pmatrix}
X^* \\ Y^* 
\end{pmatrix} .
\end{equation}
The inverse rotation reads
\begin{equation}
\begin{pmatrix}
X^* \\ Y^* 
\end{pmatrix} = 
\begin{pmatrix}
\sqrt{\lambda} &- \sqrt{1-\lambda} \\
\sqrt{1-\lambda} & \sqrt{\lambda}
\end{pmatrix}
\begin{pmatrix}
\widetilde{X} \\ \widetilde{Y} 
\end{pmatrix}
\end{equation}
which gives 
$\sqrt{\lambda} T(X^*) + \sqrt{1-\lambda} U(Y^*)=
\sqrt{\lambda} T(\sqrt{\lambda} \widetilde{X}-\sqrt{1-\lambda}\widetilde{Y})  + \sqrt{1-\lambda}U(\sqrt{1-\lambda}\widetilde{X}+\sqrt{\lambda}\widetilde{Y})$, a function of $(\widetilde{X}, \widetilde{Y})$ which we denote by $T_{\widetilde{Y}}(\widetilde{X})$.  Now since conditioning reduces entropy,
\begin{equation}\label{cond}
h(\sqrt{\lambda} X + \sqrt{1-\lambda} Y)=h(T_{\widetilde{Y}}(\widetilde{X}))\geq h(T_{\widetilde{Y}}(\widetilde{X})|\widetilde{Y}).
\end{equation}
Lemma~\ref{changevariablecond} applies with
\begin{align}
T'_{\widetilde{Y}}(\widetilde{X})&={\lambda} T'\!(\sqrt{\lambda}\widetilde{X}\!-\!\sqrt{1\!-\!\lambda}\widetilde{Y}) \!+\! (1\!-\!\lambda)U'\!(\sqrt{1\!-\!\lambda}\widetilde{X}\!+\!\sqrt{\lambda}\widetilde{Y}) \notag\\
&=\lambda T'(X^*) + (1-\lambda) U'(Y^*)
\end{align}
which gives
\begin{equation}\label{cvcond}
 h(T_{\widetilde{Y}}(\widetilde{X})|\widetilde{Y})=h(\widetilde{X}|\widetilde{Y})+
  \E \{ \log \bigl(\lambda T'(X^*) + (1-\lambda) U'(Y^*)\bigr) \}.
\end{equation}
Since $X^*,Y^*$ are independent Gaussian with identical powers, so are the rotated variables $\widetilde{X}, \widetilde{Y}$. By independence,
\begin{equation}\label{indep}
h(\widetilde{X}|\widetilde{Y})=h(\widetilde{X})=h(\sqrt{\lambda} \,X^* + \sqrt{1-\lambda}\, Y^*).
\end{equation}
Therefore, combining~\eqref{cond},~\eqref{cvcond} and~\eqref{indep} we obtain
\begin{multline}
h(\sqrt{\lambda} X + \sqrt{1-\lambda} Y) \geq
 h(\sqrt{\lambda} \,X^* + \sqrt{1-\lambda}\, Y^*) \\+  \E \{ \log \bigl(\lambda T'(X^*) + (1-\lambda) U'(Y^*)\bigr) \} .
\end{multline}
With~\eqref{transportxy} we conclude that
\begin{equation}
\begin{aligned}
&h(\sqrt{\lambda} X + \sqrt{1-\lambda} Y)  - \lambda h(X) - (1-\lambda) h(Y) \\
&\geq 
h(\sqrt{\lambda} X^*\! + \sqrt{1-\lambda} Y^*\!)  - \lambda h(X^*\!) - (1-\lambda) h(Y^*\!) \\
&\hphantom{\geq }+ \E \{ \log \bigl(\lambda T'(X^*) + (1-\lambda) U'(Y^*)\bigr) \}\\
&\hphantom{\geq }- \lambda \E \{ \log T'(X^*) \}- (1-\lambda) \E \{ \log U'(Y^*) \}
\end{aligned}\label{final}
\end{equation}
where the first line in the r.h.s. vanishes by~\eqref{hconcave*} and the remaining part is $\geq 0$ by Jensen's inequality (concavity of the logarithm). This proves the EPI~\eqref{hconcave}.\qed

\section{The Equality Case}
The equality case is easily captured by the above method. If equality holds in~\eqref{final} then 
\begin{multline}
 \log \bigl(\lambda T'(X^*) + (1-\lambda) U'(Y^*)\bigr) \\=\lambda \E \{ \log T'(X^*) \}+ (1-\lambda) \E \{ \log U'(Y^*) \text{ a.e.}
\end{multline}
Because the logarithm is strictly concave and  $0<\lambda<1$, this implies
\begin{equation}
T'(X^*)= U'(Y^*) \text{ a.e.}
\end{equation}
Since $X^*$ and $Y^*$ are independent, it follows that $T'$ and $U'$ are constant and equal, hence $T$ and $U$ are linear and $X=c\cdot X^*$, $Y=c\cdot Y^*$ are Gaussian with the same power. This is the required equality case of the EPI~\eqref{hconcave}. Of course this condition also implies equality in~\eqref{cond} since then $\sqrt{\lambda} X + \sqrt{1-\lambda} Y=T_{\widetilde{Y}}(\widetilde{X})=c\widetilde{X}$ is independent of $\widetilde{Y}$.

\section{Generalization to random vectors}

The above proof of the EPI carries over \emph{verbatim} to random vectors in $n$ dimensions. The only change is that transport maps $T:\R^n\to\R^n$ are $n$-dimensional---accordingly, $T'$ denotes the \emph{Jacobian} determinant of $T$. Lemma~\ref{transport} is easily extended to random vectors using the so-called \emph{Knöthe's map} in the theory of convex bodies~\cite{Knothe57,Villani08}, of the form
\begin{equation}
T(x)= \bigl( T_1(x_1),T_2(x_1,x_2) , \ldots, T_n(x_1,\ldots,x_n) \bigr)
\end{equation}
where $x=(x_1,x_2,\ldots,x_n)\in\R^n$. The Jacobian matrix of $T$ is triangular with positive diagonal elements:
\begin{equation}
\begin{pmatrix}
\frac{\partial T_1}{\partial x_1} & 0 & \cdots & 0\\
\frac{\partial T_2}{\partial x_1} & \frac{\partial T_2}{\partial x_2} & \cdots & 0\\
\hdotsfor{4}\\
\frac{\partial T_n}{\partial x_1} & \frac{\partial T_n}{\partial x_2} & \cdots & \frac{\partial T_n}{\partial x_n}\\
\end{pmatrix}
\end{equation}
so that
\begin{equation}
T'(x)=\prod_{i=1}^n \frac{\partial T_i}{\partial x_i} >0.
\end{equation}
This transport map was used in~\cite{RioulCosta16} and details about its construction can also be found in~\cite{Rioul17}.
Lemmas~\ref{changevariable} and~\ref{changevariablecond} are then obtained by a change of variable in $n$ dimensions. The above proof of the EPI is identical word for word, where the concavity of the logarithm in the last step~\eqref{final} is used on each dimension.

\section{A Reverse EPI}

\subsection{Derivation: Generalization to non-Gaussian $X^*$ and $Y^*$}
The above proof of the EPI can also be generalized to the case where $X^*$ and $Y^*$ are \emph{not} necessarily Gaussian. In fact a closer look at the proof reveals that the Gaussian assumption is never used except for the simplification in~\eqref{indep} which relies on the independence of $\widetilde{X}=\sqrt{\lambda} X^*\! + \sqrt{1-\lambda} Y^*$ and $\widetilde{Y}=-\sqrt{1-\lambda} X^*+\sqrt{\lambda}Y^*$. If such an independence does not hold, we obtain the more general inequality
\begin{equation}
\begin{aligned}
&h(\sqrt{\lambda} X + \sqrt{1-\lambda} Y)  - \lambda h(X) - (1-\lambda) h(Y) \\
&\geq 
h(\sqrt{\lambda} X^*\! + \sqrt{1-\lambda} Y^* |-\sqrt{1-\lambda} X^*+\sqrt{\lambda}Y^* \!)\\
&\hphantom{\geq}\qquad - \lambda h(X^*\!) - (1-\lambda) h(Y^*\!) \\
\end{aligned}
\end{equation}
valid for any independent $X,Y$ and any independent $X^*,Y^*$. In fact this gives two independent inequalities: 
For Gaussian $X^*,Y^*$ the r.h.s. vanishes and we recover the classical EPI. But for Gaussian $X,Y$ the l.h.s. vanishes, so that the r.h.s. is $\leq 0$, and we obtain a \emph{reverse} inequality which (rewritten for $X,Y$) takes the form
\begin{equation}\label{reverse}
 h(\sqrt{\lambda} X\! + \sqrt{1\!-\!\lambda} Y |-\!\sqrt{1\!-\!\lambda} X+\sqrt{\lambda}Y \!) \leq \lambda h(X\!) + (1\!-\!\lambda) h(Y\!).
\end{equation}
Compared to~\eqref{hconcave}, the opposite inequality holds but for a conditional differential entropy. In other words, $\lambda h(X\!) + (1-\lambda) h(Y\!)$ is upper bounded by the differential entropy of $\sqrt{\lambda} X + \sqrt{1-\lambda} Y$ and lower bounded by its conditional differential entropy given $-\sqrt{1-\lambda} X+\sqrt{\lambda}Y$, the difference between the bounds being equal to the mutual information $I(\sqrt{\lambda} X + \sqrt{1-\lambda} Y;-\sqrt{1-\lambda} X+\sqrt{\lambda}Y)$. Thus an equivalent restatement is
\begin{equation}
\begin{split}
0 \leq 
h(\sqrt{\lambda}\, X + \sqrt{1-\lambda}\, Y) 
-  \lambda h(X) -(1-\lambda) h(Y)  \\
\leq I( \sqrt{\lambda} \,X + \sqrt{1-\lambda}\, Y ; -\sqrt{1-\lambda}\, X + \sqrt{\lambda} \,Y ).
 \end{split}\label{deficit}
\end{equation}
This mutual information can be seen as an upper bound on the \emph{deficit} in the EPI for $X$ and $Y$, which is zero if and only if $X$ and $Y$ are Gaussian with identical powers. Courtade~\cite{Courtade17} recently derived a similar bound on the deficit in the logarithmic Sobolev inequality, which is equivalent to another type of ``reverse EPI''. As above the extension to random vectors in $n$ dimensions is straightforward.

\subsection{The Equality Case and Bernstein's Lemma}

\begin{figure*}[!t]
\normalsize
\begin{align*}
\K_{U|V} &=  \lambda \K_X + (1-\lambda) \K_Y - \lambda(1-\lambda) (\K_Y-\K_X)\bigl[(1-\lambda) \K_X + \lambda \K_Y\bigr]^{-1}(\K_Y-\K_X)
\\&= \lambda \K_X  - \lambda(1-\lambda) \K_X \bigl[(1-\lambda) \K_X + \lambda \K_Y\bigr]^{-1} \K_X 
+ \lambda(1-\lambda) \K_X \bigl[(1-\lambda) \K_X + \lambda \K_Y\bigr]^{-1} \K_Y\\
&\quad+(1-\lambda) \K_Y  - \lambda(1-\lambda) \K_Y \bigl[(1-\lambda) \K_X + \lambda \K_Y\bigr]^{-1} \K_Y
+  \lambda(1-\lambda) \K_Y \bigl[(1-\lambda) \K_X + \lambda \K_Y\bigr]^{-1} \K_X\\
&= \lambda \K_X\bigl[(1-\lambda) \K_X + \lambda \K_Y\bigr]^{-1} \bigl( (1-\lambda) \K_X + \lambda \K_Y - (1-\lambda) \K_X\bigr) + \lambda(1-\lambda) \K_X \bigl[(1-\lambda) \K_X + \lambda \K_Y\bigr]^{-1} \K_Y\\
&\quad+ (1-\lambda) \K_Y\bigl[(1-\lambda) \K_X + \lambda \K_Y\bigr]^{-1} \bigl( (1-\lambda) \K_X + \lambda \K_Y - \lambda \K_Y\bigr)+\lambda(1-\lambda) \K_Y \bigl[(1-\lambda) \K_X + \lambda \K_Y\bigr]^{-1} \K_X\\
&=
\bigl(\lambda^2 + \lambda(1-\lambda) + (1-\lambda)^2 + \lambda(1-\lambda)\bigr)
\bigl[\lambda \K_X^{-1} + (1-\lambda) \K_Y^{-1} \bigr]^{-1} = \bigl[\lambda \K_X^{-1} + (1-\lambda) \K_Y^{-1} \bigr]^{-1}.
\end{align*}
\hrulefill
\end{figure*}

We have seen that equality holds in the above proof of the EPI if and only if $(X,Y)$ and $(X^*,Y^*)$ are proportional. The same argument shows that the the same equality condition holds for the reverse EPI. Thus both the EPI~\eqref{hconcave} and its reverse~\eqref{reverse} are equalities if and only if $X$ and $Y$ are i.i.d. Gaussian. This also corresponds to the case where the mutual information vanishes in~\eqref{deficit}. This gives an alternative proof of Bernstein's lemma (see e.g.,~\cite[Appendix~I]{GengNair14} and~\cite[Chap.~5]{Bryc95}):
\begin{lemma}[Bernstein]
Let $X$ and $Y$ be independent. Then the rotated $\sqrt{\lambda} \,X + \sqrt{1-\lambda}\, Y$, $-\sqrt{1-\lambda}\, X + \sqrt{\lambda} \,Y$ are independent if and only if $X$, $Y$ are i.i.d. Gaussian.
\end{lemma}

\subsection{The Gaussian Case}

If $X$ and $Y$ are Gaussian with not necessarily equal powers $P(X)$ and $P(Y)$, it is easily seen that~\eqref{reverse} and~\eqref{hconcave} reduce to the harmonic/geometric/arithmetic inequalities
\begin{multline}
\bigl(\lambda P(X)^{-1} + (1-\lambda) P(Y)^{-1}\bigr)^{-1} \\\leq P(X)^\lambda P(Y)^{1-\lambda}
\\\leq \lambda P(X)+(1-\lambda) P(Y).
\end{multline}
More generally for Gaussian vectors, if $X\sim\mathcal{N}(0,\K_X)$ and $Y\sim\mathcal{N}(0,\K_Y)$ not necessarily of identical covariances, it is known~\cite[Thm.~8]{DemboCoverThomas91} that the EPI reduces to Ky Fan's concavity inequality of the log-determinant: using the well-known formula $h(U)=\frac{1}{2}\log \bigl((2\pi e)^n |\K_U|\bigr)$ the EPI reduces to
$\log |\lambda \K_X + (1-\lambda) \K_Y| \geq \lambda \log|\K_X| + (1-\lambda) \log |\K_Y|$. 

Similarly for the reverse EPI, noting that $h(U|V) = \frac{1}{2}\log \bigl((2\pi e)^n |\K_{U|V}|\bigr)$ where 
 $\K_{U|V}$ is  Schur's complement $\K_{U|V}=\K_U - \K_{UV}\K_V^{-1} \K_{VU}$  (where $\K_{UV}$ is an intercovariance matrix), set $U=\sqrt{\lambda} \,X + \sqrt{1-\lambda} \,Y$ and $V=-\sqrt{1-\lambda} \,X + \sqrt{\lambda}\, Y$, $\K_U=\lambda \K_X + (1-\lambda) \K_Y$, $\K_V=(1-\lambda) \K_X + \lambda \K_Y$, and $\K_{UV}= \K_{VU}= \sqrt{\lambda(1-\lambda)} (\K_Y-\K_X)$. By the calculation shown at the top of this page, the reverse EPI  reduces to the inequality
$\log |  \lambda \K_X^{-1} + (1-\lambda) \K_Y^{-1} |^{-1} \leq \lambda \log|\K_X| + (1-\lambda) \log |\K_Y|$.
Thus ~\eqref{reverse} and~\eqref{hconcave} reduce to the generalized harmonic/geometric/arithmetic inequalities:
\begin{multline}
|  \lambda \K_X^{-1} + (1-\lambda) \K_Y^{-1} |^{-1} \\
\leq |\K_X|^\lambda |\K_Y|^{1-\lambda} 
\\\leq  |\lambda \K_X + (1-\lambda) \K_Y|.
\end{multline}

\subsection{Equivalence Between the EPI and its Reverse}  

As observed by Chandra Nair in a private communication to the author, it turns out that the reverse EPI is in fact equivalent to the EPI where the roles of $X$ and $Y$ are permuted. In fact~\eqref{reverse} is equivalent to
\begin{equation}
\begin{aligned}
h( \sqrt{\lambda} \,X &+ \sqrt{1-\lambda}\, Y , -\sqrt{1-\lambda}\, X + \sqrt{\lambda} \,Y ) \\
&\leq \lambda h(X) +(1-\lambda) h(Y) + h(-\sqrt{1-\lambda}\, X + \sqrt{\lambda} \,Y)
\end{aligned} 
\end{equation}
where the joint entropy in the r.h.s. equals $h(X,Y)=h(X)+h(Y)$ by the scaling property of the differential entropy for vectors. Reorganizing terms one obtains the following version of the EPI:
\begin{equation}
 (1-\lambda) h(X)+\lambda h(Y) \leq  h(-\sqrt{1-\lambda}\, X + \sqrt{\lambda} \,Y).
\end{equation}
We recover, in particular, that the cases of equality are the same for the EPI and its reverse. The above calculation was already used by Wang and Madiman~\cite{WangMadiman13} as a short proof of the EPI under the hypothesis that $X$ and $Y$ follow symmetrical and identical distributions. One reason why the EPI is equivalent to its reverse version is suggested below in relation to Young's convolutional inequality and its reverse.

\section{Zamir and Feder's Generalization to Linear Transformations}

An immediate generalization of the EPI~\eqref{hconcave} for $n$ independent variables $X_1,X_2,\ldots,X_n$ is
\begin{equation}
h\bigl(\sum_i a_i X_i \bigr) \geq \sum_i a_i^2 h(X_i) 
\end{equation}
where the coefficients are normalized such that $\sum_i a_i^2=1$. The above proof of the EPI can easily be adapted to prove this inequality directly by letting $\mathbf{A}$ be an orthogonal matrix whose first line is $(a_1,a_2,\ldots,a_n)$ and defining 
\begin{equation}
\widetilde{X} = \mathbf{A} {X^*}
\end{equation}
where $\widetilde{X^*}$ is a column vector of $n$ i.i.d. Gaussian variables $X^*_1,X^*_2,\ldots,X^*_n$.
The inverse transformation is $X^*=\mathbf{A\!}^t \widetilde{X}$ and the proof is easily modified along these lines.

Essentially the same proof can be used for Zamir and Feder's generalized EPI~\cite{ZamirFeder93} (see also~\cite[\S~IV]{Rioul11}):
\begin{equation}
h(\mathbf{A}X) \geq \sum_{i,j} a^2_{i,j} h(X_j) 
\end{equation}
where $X$ is the column vector of components $X_1,X_2,\ldots,X_n$ and $\mathbf{A}=(a_{i,j})$ is any real-valued (possibly rectangular) matrix with orthonormal rows. By adding orthonormal rows we form a square orthonormal matrix (still denoted by $\mathbf{A}$) and the same transformation $\widetilde{X} = \mathbf{A} {X^*}$, $X^*=\mathbf{A\!}^t \widetilde{X}$ is used. The conclusion follows from a simple inequality~\cite[Lemma~1]{GuoShamaiVerdu06} which generalizes Jensen's inequality for the logarithm.

\section{Generalization to R\'enyi Entropies and the Relation to Young's Inequality}

The R\'enyi entropy of order $p>0$ ($p\ne 1$) is defined as
\begin{equation}
h_p(X)=h_p(f)=\frac{1}{1-p}\log\int f^p = -p'\log \|f\|_p
\end{equation}
where $\|f\|_p$ denotes the $L^p$ norm of the density $f$ of $X$ and $p'$ is $p$'s conjugate such that $1/p+1/p'=1$. 
While the above proof of the EPI focuses on the equivalent inequality
\begin{equation}
\begin{aligned}
h&(\sqrt{\lambda} X + \sqrt{1\!-\!\lambda}\, Y) - \lambda h(X) -(1\!-\!\lambda) h(Y)\\
&\geq h(\sqrt{\lambda} X^* + \sqrt{1\!-\!\lambda}\, Y^*\!) - \lambda h(X^*\!) -(1\!-\!\lambda) h(Y^*\!) 
\end{aligned}
\end{equation}
for i.i.d. Gaussian $X^*,Y^*$, the natural generalization of the EPI considered in~\cite{DemboCoverThomas91} takes the form
\begin{equation}
\begin{aligned}
h&_r(\sqrt{\lambda} X + \sqrt{1\!-\!\lambda}\, Y) - \lambda h_p(X) -(1\!-\!\lambda) h_q(Y)\\
&\geq h_r(\sqrt{\lambda} X^* \!\!+\! \sqrt{1\!-\!\lambda}\, Y^*\!) \!-\! \lambda h_p(X^*\!) \!-\!(1\!-\!\lambda) h_q(Y^*\!) 
\end{aligned}\label{renyi}
\end{equation}
where $p,q,r$ are chosen such that $1/p'=\lambda/r'$ and $1/q'=(1-\lambda)/r'$ (so that $1/p'+1/q'=1/r'$), where $p'$, $q'$, and $r'$ are the conjugates of $p$, $q$ and $r$, respectively.

The EPI can then be obtained by letting $p,q,r\to 1$ as shown in~\cite{DemboCoverThomas91}. This is the preferred proof of the EPI in the classical textbook by Cover and Thomas~\cite{CoverThomas06}. Notice that because the EPI is obtained as a limit, the equality case is not settled by this method.

The above transportation proof can also be generalized, with the same transport maps $X=T(X^*)$ and $Y=U(Y^*)$, to prove~\eqref{renyi}. Now the R\'enyi EPI~\eqref{renyi} is in fact equivalent to the sharp Young's convolutional inequality or its reverse inequality~\cite{DemboCoverThomas91,CoverThomas06}. This is easily seen by noting that the 
R\'enyi entropy enjoys the same scaling property~\eqref{scaling} as the differential entropy  so that~\eqref{renyi} is equivalent to saying that
$h_r(\sqrt{\lambda} X + \sqrt{1\!-\!\lambda}\, Y) - \lambda h_p(\sqrt{\lambda}X) -(1\!-\!\lambda) h_q(\sqrt{1\!-\!\lambda}\, Y)$ (which equals $-r'\log\|f\ast g\|_r + r'\log \|f\|_p + r'\log \|g\|_q$ where $f$ and $g$ denote the densities of $\sqrt{\lambda}X$ and $\sqrt{1-\lambda}Y$) is minimized for i.i.d. Gaussian $X$ and $Y$. Dividing by $r'$ gives sharp Young's inequality for $p,q,r>1$ ($r'>0$):
\begin{equation}
\sqrt{\frac{r^{1/r}}{r'^{1/r'}}}\| f\ast g\|_r  \leq \sqrt{\frac{p^{1/p}}{p'^{1/p'}}}\|f\|_p \cdot \sqrt{\frac{q^{1/q}}{q'^{1/q'}}}\|g\|_q 
\end{equation}
and the \emph{reverse} Young's inequality for  $0<p,q,r<1$ ($r'<0$):
\begin{equation}
\sqrt{\frac{r^{1/r}}{|r'|^{1/r'}}}\| f\ast g\|_r  \geq \sqrt{\frac{p^{1/p}}{|p'|^{1/p'}}}\|f\|_p \cdot \sqrt{\frac{q^{1/q}}{|q'|^{1/q'}}}\|g\|_q .
\end{equation}
In fact Barthe~\cite{Barthe98} gave a transportation proof of both inequalities. Since one obtains the EPI by letting $p,q,r\to 1^+$ from above (from Young's inequality) and also by letting $p,q,r\to 1^-$ from below (from the reverse Young's inequality), the EPI and its reverse are equivalent at the limit $p,q,r\to 1$.

\section*{Acknowledgment}
The author would like to thank Tom Courtade, Chandra Nair and Mich\`ele Wigger for their discussions.



\vspace*{0ex}

\end{document}